\documentclass[conference]{IEEEtran}
\IEEEoverridecommandlockouts
\pdfoutput=1
\usepackage{wrapfig}
\usepackage{cite}
\usepackage{amsmath,amssymb,amsfonts,amsthm}
\usepackage[noend]{algpseudocode}
\usepackage{algorithm}
\usepackage{algorithmicx}
\usepackage{graphicx}
\usepackage{hyperref}
\usepackage[capitalise]{cleveref}
\usepackage{mathtools}
\usepackage{textcomp}
\usepackage{xcolor}
\usepackage{tikz}
\usepackage[T1]{fontenc}
\usepackage{numprint}

\newif\ifmicrotype

\microtypetrue

\ifmicrotype
\usepackage[%
activate={true,nocompatibility},%
final,%
tracking=true,%
kerning=true,%
spacing=true,%
stretch=10,%
shrink=30]{microtype}
\microtypecontext{spacing=nonfrench}
\SetTracking{encoding={*}, shape=sc}{0}
\fi

\usepackage[per-mode=symbol,binary-units=true]{siunitx}
\DeclareSIUnit{\nothing}{\relax}
\usepackage{collcell}
\usepackage{multirow}
\usepackage{booktabs}
\DeclareSIPrefix\billion{B}{9}

\def\BibTeX{{\rm B\kern-.05em{\sc i\kern-.025em b}\kern-.08em
T\kern-.1667em\lower.7ex\hbox{E}\kern-.125emX}}
\newcommand\Set[2]{\{\,#1\mid#2\,\}}
\newcommand{\etal}{et~al.}
\newcommand{\ie}{i.e.,}
\newcommand{\eg}{e.g.,}
\newcommand{\Boruvka}{Bor\r{u}vka}

\newcommand{\RMAT}{RMAT}
\newcommand{\usroad}{\texttt{US-road}}
\newcommand{\friendster}{\texttt{friendster}}
\newcommand{\twitter}{\texttt{twitter}}

\newcommand{\wcd}{\texttt{wdc-14}}
\newcommand{\friendsterN}{friendster}
\newcommand{\twitterN}{twitter}
\newcommand{\ukwebN}{uk-2007}
\newcommand{\itwebN}{it-2004}
\newcommand{\wcdN}{wdc-14}
\newcommand{\usroadN}{US-road}

\newcommand{\allgather}{allgather}

\newcommand{\alltoall}{all-to-all}

 	\definecolor{emerald}{rgb}{0.31, 0.78, 0.47}

\DeclarePairedDelimiter\parentheses{\lparen}{\rparen}

\newcommand{\src}[1]{\operatorname{src} \parentheses*{#1}}
\newcommand{\dst}[1]{\operatorname{dst} \parentheses*{#1}}

\newcommand{\LandauO}[1]{\mathcal{O}( #1 )}

\newcommand{\column}[1]{\operatorname{col} \parentheses*{#1}}
\newcommand{\row}[1]{\operatorname{row} \parentheses*{#1}}

\newcommand{\lexmin}[1]{\operatorname{min}_{\text{lex}}\parentheses*{#1}}

\DeclarePairedDelimiter\ceil{\lceil}{\rceil}
\DeclarePairedDelimiter\floor{\lfloor}{\rfloor}

\algnewcommand{\Downto}{\textbf{ downto }}
\algnewcommand{\By}{\textbf{ by }}

\newtheorem{theorem}{Theorem}

\newif\ifieeecopyright
\ieeecopyrighttrue

\ifieeecopyright
\makeatletter
\def\ps@IEEEtitlepagestyle{%
  \def\@oddfoot{\mycopyrightnotice}%
  \def\@evenfoot{}%
}
\def\mycopyrightnotice{%
    {\footnotesize 
        \begin{minipage}{\textwidth}
        \textcopyright~2023 IEEE. Personal use of this material is permitted. Permission
from IEEE must be obtained for all other uses, in any current or future
media, including reprinting/republishing this material for advertising or
promotional purposes, creating new collective works, for resale or
redistribution to servers or lists, or reuse of any copyrighted
component of this work in other works. Published version: \href{https://doi.org/10.1109/IPDPS54959.2023.00075}{10.1109/IPDPS54959.2023.00075}~\cite{SandersS23}.
        \end{minipage}
}
  \gdef\mycopyrightnotice{}
}
\makeatother
\fi

\begin{document}

\nocite{SandersS23}

\title{Engineering Massively Parallel MST Algorithms
	\thanks{
		\begin{wrapfigure}{r}{.33\columnwidth}
			\vspace{-\baselineskip}
			\includegraphics[width=.33\columnwidth]{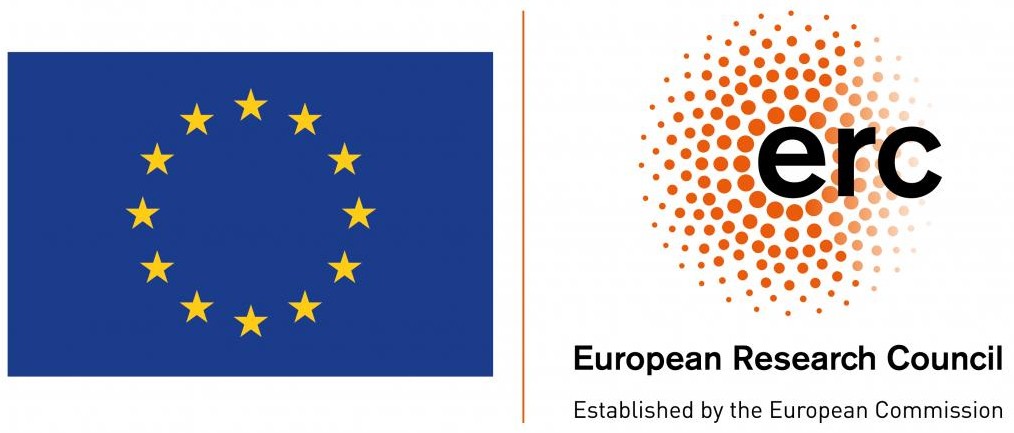}
		\end{wrapfigure}
		This project has received funding from the European Research Council (ERC) under the European Union’s Horizon 2020 research and innovation programme (grant agreement No. 882500).
	}
}

\author{\IEEEauthorblockN{ Peter Sanders}
	\IEEEauthorblockA{\textit{Institute of Theoretical Informatics} \\
		\textit{Karlsruhe Institute of Technology}\\
		Karlsruhe, Germany \\
		sanders@kit.edu}
	\and
	\IEEEauthorblockN{Matthias Schimek}
	\IEEEauthorblockA{\textit{Institute of Theoretical Informatics} \\
		\textit{Karlsruhe Institute of Technology}\\
		Karlsruhe, Germany \\
		schimek@kit.edu}

}

\maketitle

\begin{abstract}
	We develop and extensively evaluate highly scalable distributed-memory
	algorithms for computing minimum spanning trees (MSTs). At the heart of our
	solutions is a scalable variant of \Boruvka's algorithm. For partitioned graphs
	with many local edges we improve this with an effective form of
	contracting local parts of the graph during a preprocessing step. We also adapt
	the filtering concept of the best practical sequential algorithm to develop a
	massively parallel Filter-\Boruvka\ algorithm that is very useful for graphs with
	poor locality and high average degree. Our experiments indicate that our
	algorithms scale well up to at least 65\,536 cores and are up to 800 times
	faster than previous distributed MST algorithms.
\end{abstract}


\section{Introduction}

Graphs are a universal way to model relations between objects and are thus
prominently needed everywhere in computing. Hence, it is not surprising that
processing huge graphs is eminently interesting for parallel processing. For
example, this interest is reflected in the graph500 benchmark
(\url{graph500.org}) that has led to the development of massively scalable
algorithms for breadth-first traversal of \RMAT\ random graphs. However, much
less work has been done on other graph problems or even other graph families.

This paper wants to change this for the fundamental problem of computing
minimum spanning trees (MSTs). Given an undirected and connected graph with weighted edges,
the MST problem asks for a subset of the  edges that connects all vertices and has
minimum total weight among all such sets.

In the search for graph problems which can be solved massively in parallel,
the MST problem is a prime candidate since it is known to have
scalable PRAM algorithms (e.g. \cite{awerbuch1987new_short,
	ColeKT96_short}) with time complexity (poly)logarithmic in the
number of vertices of the input graph. The MST problem is not only of
theoretical interest, but also has many applications, \eg{} clustering, image
segmentation, and network design \cite{bateni2017affinity_short,
	wassenberg2009efficient_short, li2005design_short}.

After discussing basic terms and tools in \cref{s:preliminaries} and previous
work in \cref{subsection_related_work}, we describe a scalable distributed
variant of \Boruvka's algorithm \cite{boruvka1926jistem_short} in
\cref{s:algorithm}.  \Boruvka{}'s algorithm is a promising candidate for the distributed setting as it is conceptually simple and has
a high potential for parallelism.

Our implementation represents the graph as a distributed
sequence of edges that is lexicographically sorted. This approach has a
large design space that we explored in an experimental, performance-driven way
with theoretical analysis as a background consideration. Essentially, a
\Boruvka\ round is reduced to a number of sparse \alltoall{}~communications
which are notoriously difficult to scale on large distributed machines due to
contention and high message-startup overheads. We mitigate this by using a
two-level sparse \alltoall~on a logical grid and a fast distributed sorter. We
also switch to a base case with replicated vertex set as soon as the number of
vertices is small enough.

Many graphs have a numbering of the vertices that assigns significant parts of a vertex's neighborhood  to the same processing element (PE) yielding \emph{local edges}.
In \cref{subsection_local_kernel}, we  develop a preprocessing algorithm that contracts local edges in such a way that the only remaining vertices have nonlocal incident edges that are lighter than any of their local incident edge. This reduces processing time by up to a factor $5$.

A disadvantage of \Boruvka's algorithm is that it may have to process all the
edges a logarithmic number of times. In \cref{section_filter_mst}, we therefore
present the Filter-\Boruvka~algorithm that combines \Boruvka's algorithm with
the \emph{filtering} approach of the Filter-Kruskal MST algorithm
\cite{osipov2009filter_short} that in many respects is currently the best
practical sequential algorithm. The idea is to first compute the minimum
spanning forest $F$ of the globally lightest edges, then drop edges that are
within components of $F$, and only then compute the MST of the graph that
includes the surviving edges. We prove that is has work linear in the number of
edges and polylogarithmic span.

The implementation outlined in \cref{s:implementation} uses hybrid parallelism
with multiple OpenMP threads operating within each process of the message
passing interface (MPI) which turns out to be crucial for being able to support
a large number of cores for many inputs.

In \cref{s:experiments}, we discuss experiments on a supercomputer
using up to $2^{16}$ cores.  This includes weak scaling experiments
using 6 families of graphs (grid, 2D/3D random geometric, hyperbolic,
Erd\H{o}s-Renyi, and RMAT). Our best algorithm scales on all these families all the way to $2^{16}$ cores
with a potential for considerable further scaling on the families that have some locality.
In many cases our implementation is one or two orders of magnitude faster than
implementations of previous distributed MST algorithms.

We also perform strong scaling experiments on six large real world
graphs with good scaling up to $2^{14}$ cores on the largest available
inputs as well as order-of-magnitude speedups over the implementations of previous algorithms.

In \cref{s:conclusion}, we summarize our result and discuss possible future improvements.

\paragraph*{Summary of Contributions}
\begin{itemize}
	\item Design and analysis of
	      a practical and highly scalable distributed variant of \Boruvka's MST algorithm.
	\item Design and analysis of a Filter-\Boruvka\ algorithm
	      with constant expected work per edge and polylogarithmic span.
	\item Scalable sparse all-to-all communication using indirect communication.
	\item Fast base case with replicated vertex set and communication-efficient preprocessing by contracting local MST edges.
	\item Hybrid implementation effectively using multithreading on each compute node.
    \item Extensive experimental evaluation on up to $2^{16}$ cores on a large variety of synthetic and real-world instances.
\end{itemize}

\section{Preliminaries}
\label{s:preliminaries}
\subsection{Machine Model and Communication Primitives}\label{ss:model}
In this work, we assume a distributed-memory computing model consisting of $p$ processing elements (PE) numbered  $0..p-1$ \footnote{$a..b$ is shorthand for the sequence $[a, a+1, \dots, b-1, b]$}
allowing single-ported point-to-point communication between arbitrary communication partners \cite{sanders2019sequential}.
In this model, it takes $\alpha + \beta \ell$ time to send a message of length $\ell$
between two PEs. The parameter $\alpha$ denotes the startup overhead to initiate a
message, while $\beta$ quantifies the time needed to communicate a data unit.

There are algorithms for the collective operations broadcast, (all)reduce and
prefix-sum with a time complexity in $\LandauO{\alpha \log(p) + \beta \ell}$.
For the \allgather{}-operation (a.k.a as \alltoall~broadcast or gossiping) we get the same complexity when $\ell$ is the sum of all message lengths.

In addition to these operations with very clear-cut complexity, in distributed graph algorithms we need more complex primitives where the complexity highly depends on the implementation and the concrete communication patterns.
(Personalized, sparse) \alltoall~communication delivers arbitrary sets of messages in a batched way.
This can be achieved in $\LandauO{\alpha p + \beta \ell}$ when $\ell$ is the bottleneck communication volume, i.e., the maximum amount of data sent or received by a PE.
The large startup term $\alpha p$ can be reduced at the cost of more and more indirect data delivery.
For example, using a hypercube communication scheme, $\LandauO{(\alpha + \beta \ell)\log(p)}$ can be achieved. In our implementations, we go part of this way and reduce the startup term to $\alpha \sqrt{p}$
by adding one indirection to the communication.

Similarly, comparison based sorting of $k$ elements can be done in expected time $\LandauO{(k\log(k)+\beta k)/p + \alpha p}$, e.g., by sample sort and direct data delivery, or in time $\LandauO{(k\log(k)+\beta k\log(p))/p + \alpha \log^2(p)}$
using an algorithm that moves the data a logarithmic number of times. Once more, an intermediate solution
moving data a constant number of times turns out to be optimal for large data sets on large distributed machines \cite{axtmann2017robust_short}.

Since our algorithms crucially depend on the performance of sparse \alltoall~and sorting,
we therefore use several implementations depending on the number of PEs used and the amount of data involved.

\subsection{Graph and Input/Output Format}\label{ss:format}

The input is an undirected weighted graph $G = (V,E)$ with vertex \emph{labels} in $1..|V|$. We use $n = |V|$ and $m = |E|$.
We say that $G$ is connected if for all vertices $u,v \in V$, $G$ contains a path from $u$ to $v$.
For a connected graph, the output is a minimum spanning tree (MST).
If $G$ consists of multiple connected components, the output is an MST for each component. Such a set of MSTs is called a minimum spanning \emph{forest} (MSF).

We assume $G$ to be represented
as a lexicographically sorted sequence of directed edges $e = (u, v, w)$ with source vertex $\src{e} = u$, destination vertex $\dst{e} = v$ and weight $w$.
For each edge $(u,v, w) \in E$ the \emph{back} edge $(v,u, w)$ is also in $E$. Note that we sometimes omit the edge weight writing $(u,v)$ instead of $(u,v,w)$.
The edge sequence $E$ is \emph{1D-partitioned} among
the $p$ PEs, \ie{} PE $i$ obtains a subsequence $E_i$ of $E$ of
size $E/p$.
By $\lexmin{E'}$ we denote the lexicographically smallest edge in $E' \subseteq E$. Lexicographical means that we sort with respect to source vertex, destination vertex and then edge weight.
We refer to the local input to PE $i$ as $G_i = (V_i, E_i)$ with
$V_i = \Set{\src{e}}{e \in E_i}$.
$V_{i}$ may share its first vertex with $V_{i-1}$ and its last vertex with $V_{i+1}$.
Such a vertex is called a \emph{shared} vertex.
For vertices $v \in V_i$ that are not shared and edges $e \in E_i$, we refer to PE $i$ as the \emph{home} PE of $v$ and $e$.

The following definitions are from the point of view of a specific PE $i$:
A vertex $v \in V_i$ is \emph{local} to PE $i$.
A vertex that is not local and appears in $E_i$ is a \emph{ghost} vertex.
An edge $e = (u,v, w)$ is \emph{local} if $\src{e}$ and $\dst{e}$ are both local.
Otherwise $e$ is a \emph{cut}-edge.
\cref{fig:vertex_type_visualisation} visualizes the different vertex and edge types.
Let $E' \subseteq E$ and $V' = \{v \in V \mid \exists e \in E': v = \src{e}  \vee v = \dst{e}\}$. Then $G' = (V',E')$ is the subgraph of $G$ induced by $E'$.

We replicate an array of size $p$ containing $\lexmin{E_i}$ for $0 \le i < p$ on each PE as part of our distributed graph data structure. This allows localization of the home PE of a vertex or edge by binary search.

\begin{figure}
	\centering
	\includegraphics[scale=0.8]{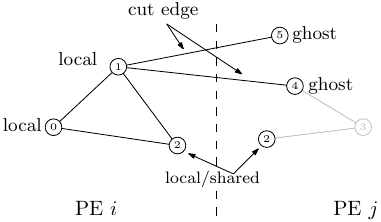}
	\caption{Visualization of the various vertex types from the point of view of PE $i$. Vertices and edges in grey are not present on PE $i$.}
	\label{fig:vertex_type_visualisation}
\end{figure}

For each MSF edge $\{u, v\}$ with weight $w$, we either find $e =
	(u,v,w)$ or $e' = (v, u, w)$ on their respective home PE as output.

\subsection{MST Properties and \Boruvka{}'s Algorithm}
\label{subsection_boruvkaMST}
Since our distributed MST algorithms conceptually follow \Boruvka{}'s
algorithm \cite{boruvka1926jistem_short}, we outline its structure in a sequential setting and state some properties of MSTs we will use later on. Let $G = (V,E)$ be
the input graph. For a graph with multiple connected components, \Boruvka{}'s algorithm (and our distributed variants) can be applied to each component independently yielding an MSF.
We can therefore assume $G$ to be connected without loss of generality.
Since one can use vertex labels to consistently break ties for equal edge weights, we also can assume our input graphs to have distinct edge weights resulting in a unique MST \cite{sanders2019sequential}.
\Boruvka{}'s algorithm runs in
so called (\Boruvka{}) rounds:
First, for each vertex the lightest incident edge is determined. These edges induce connected components in $G$. These components are trees with one 2-cycle, \ie{} one additional edge $(u,v)$, so called \emph{pseudo trees} \cite{sanders2019sequential}.
By some tie breaking rule,  such a pseudo tree can be converted into a rooted tree $T$ with root $u$.
By the min-cut property, all edges of $T$ are MST edges \cite{sanders2019sequential}.
These trees are then contracted into a single vertex -- the so called
\emph{component root}, for which the vertex label of the tree root $u$ is used and edges are relabeled accordingly. Subsequently,
self-loops are deleted. Optionally, parallel edges can be deleted keeping only the
lightest among these. This results in a new graph $G' = (V', E')$ with $V'$
being the set of the component roots on which we subsequently apply a \Boruvka{} round until $|V'| = 1.$
Since we find $|V'| \le |V| / 2$, the computation finishes after at most $\log(|V|)$ rounds.
Note that we do not necessarily need to contract the rooted trees $T$ into
exactly one component. Instead we can also split $T$ into multiple subtrees and
define $V'$ as the set of the roots of all such subtrees.
We will use this property in our distributed MST algorithms to handle shared vertices.

\section{Related Work}
\label{subsection_related_work}
The computation of minimum spanning trees has been extensively studied in many
different models of computation. The earliest MST algorithms in the sequential
setting are due to \Boruvka{} \cite{boruvka1926jistem_short}, Jarnik-Prim
\cite{jarnik1930_short} and Kruskal \cite{kruskal1956shortest_short}.
Many of the later developed algorithms use ideas from these three algorithms.
The KKT-algorithm \cite{karger1995randomized_short} combines \Boruvka's algorithm
with randomized sampling and edge filtering to obtain an algorithm with
linear expected running time.
The Filter-Kruskal algorithm \cite{osipov2009filter_short} combines Kruskal's algorithm with a simplified approach
to edge filtering that works provably well for random edge weights and arguably is the fastest practical sequential algorithm.

In the PRAM model, there are multiple algorithms with polylogarithmic span (critical path length of the computation DAG) \cite{awerbuch1987new_short, chong2003improving_short}.
More practical algorithms have been devised for the shared-memory (single-node)
parallel setting \cite{dhulipala2021theoretically_short, bader2006fast_short, bader2005fast_short, osipov2009filter_short}.
The shared-memory \Boruvka-variant described in \cite{dhulipala2021theoretically_short}
has some similarities to our distributed Filter-\Boruvka\ described in \cref{section_filter_mst}
but uses iterative skewed partitioning rather than recursive symmetric partitioning.
We are also not aware of an analysis comparable to ours.
Very recently, Esfahani~\etal~\cite{esfahani2022mastiff_short} proposed a \emph{structure-aware} algorithm
that outperforms the previous shared-memory state-of-the-art MST algorithm~\cite{dhulipala2021theoretically_short} for graphs with skewed degree distributions.
We discuss their results in \cref{s:experiments}.
For an overview of GPU algorithms, we refer to the related work section of \cite{panja2018mnd_short}.

External MST algorithms
\cite{arge2002cache_short,dementiev2004engineering_short} take memory locality
into account but have inherently sequential components. An MST algorithm for
the parallel external memory model (PEM) \cite{argeGS10_short} works in a small
number of rounds using Euler tour construction and list-ranking as subroutines.
We did not follow that approach as it is relatively complicated and does not
address contention and constant factors in the number of synchronizations in
the way we need it.

Finding MSTs has also been extensively examined in distributed computing
models. Chung and Condon \cite{chung1996parallel_short} were the first to present a
distributed implementation of \Boruvka{}'s algorithm. They use distributed
pointer doubling to contract components and evaluate their algorithm on up to
$64$ cores.
Dehne~\etal~\cite{dehne1998practical_short} provide several practical algorithms and
an evaluation on up to $16$ cores for dense graphs with $m/n > p$. This
condition essentially ensures that the memory of each machine is large enough
to hold all vertices of the graph.
Loncar~\etal~\cite{loncar2013distributed_short} propose distributed variants
of the Kruskal and Jarnik-Prim algorithm that also rely on replicated vertices.

Many distributed graph processing frameworks \cite{salihoglu2013gps_short,
	apacheGiraph_short, yan2015effective_short, li2018regraph_short,panja2018mnd_short} include a (mostly \Boruvka-based) MST algorithm to
evaluate their performance.
ReGraph \cite{li2018regraph_short} and  MND-MST \cite{panja2018mnd_short}
exploit locality by MST computations on locally available subgraphs.
However, none of these approaches have been evaluated on more
than $400$ cores.
Furthermore, many of the above-mentioned graph processing frameworks need considerable time
to load and prepare/partition a graph before the actual computation starts \cite{han2014experimental_short}.
In \cref{s:experiments}, we directly compare ourselves with a version of MND-MST.
Since MND-MST outperforms other graph tools in previous studies \cite{panja2018mnd_short},
this also allows some transitive comparisons.

The MST problem has also received considerable attention in MapReduce/MPC models \cite{karloff2010model_short,
	lattanzi2011filtering_short, qin2014scalable_short,  bateni2017affinity_short, andoni2018parallel_short,
	behnezhad13parallel_short, behnezhad2021massively_short}. Algorithms in these models aim
at the reduction of communication rounds while ensuring that the respective
model's upper bound on communication volume per round is met. In
\cite{karloff2010model_short, lattanzi2011filtering_short, qin2014scalable_short}, the proposed
algorithms use similar ideas as previous parallelizations of \Boruvka's algorithm \cite{chung1996parallel_short,dehne1998practical_short}.
While \cite{karloff2010model_short} and \cite{lattanzi2011filtering_short} present only theoretical results,
the practical implementation of \cite{qin2014scalable_short} on up to $68$ cores reveals
scalability problems whereas our algorithms scale up to $2^{14}$ cores on the same graphs and are much faster when using a comparable number of cores.

Better performance than plain MapReduce can be achieved by adding a distributed hash-table
\cite{bateni2017affinity_short,behnezhad2021massively_short, behnezhad13parallel_short}.
An implementation \cite{behnezhad13parallel_short} using between $200$ and $3\,600$ cores
reports running times on the friendster and twitter graph, which are also part of our benchmark set.
Unfortunately, the number of cores actually used is not given.
Our code on $256$ cores is $37$ times faster on friendster and $77$ times faster on twitter.
Using $4096$ cores, the speedup grows to $297$ and $490$, respectively.

Recently, Baer~\etal\cite{baer2022parallel_short} used sparse matrix kernels to adapt the
Awerbuch-Shiloach-PRAM algorithm \cite{awerbuch1987new_short} to the distributed setting leveraging a distributed library for sparse tensor algebra.
To our knowledge this work was the first to evaluate the scalability of an MST
algorithm on systems with more than a few hundred to thousand cores using
up to $17k$ cores. Therefore, we include their algorithm in our experimental evaluation.

Large scale experiments using up to 6000 cores have also been done for complete
graphs stemming from geometric MST-based clustering problems
\cite{olman2008parallel_short, hendrix2013scalable_short, goyal2016fast_short}.
However, processing dense graph is much easier to parallelize and the algorithms used there
would not scale for the large sparse graphs (with three orders of magnitude more vertices) we are considering.

\section{Scalable Distributed-Memory \Boruvka{}}\label{s:algorithm}
Algorithm \ref{alg:basic_msf} shows pseudocode for our distributed \Boruvka-MST algorithm where PE $i$ works on
the local subgraph $G_i$.
We give a
high level overview before describing subroutines in more
detail.
We first exploit
locality by determining local MST edges $T_i$ without communication (see \cref{subsection_local_kernel}).
These edges are then contracted yielding smaller subgraphs $G_i$ which have to be
processed in a distributed manner.
In every distributed \Boruvka{} round, each PE $i$ first determines the lightest edges incident to its local vertices $E_i^{\mathrm{min}}$. Shared vertices are only considered in the base case discussed below.
This simplifies several parts of the algorithm and has no significant effect on the running time
as it can be shown that the number of local vertices shrinks by at least a factor of
two in every round.
The components induced by the edges $E_i^{\mathrm{min}}$ are then contracted and the component
roots $L^{\mathrm{local}}_i$ are determined (see \cref{subsection_contraction_label_exchange}).
Afterwards, each PE retrieves the new labels for its ghost vertices $L^{\mathrm{ghost}}_i$.
Using this information, the contracted graph is built using the operations {\sc Relabel} and {\sc Redistribute}
eliminating self-loops and parallel edges on the way (see \cref{subsection_redistribution}).

As soon as the global number of vertices is small enough to be stored on one PE, we
switch to our base case algorithm (see \cref{subsection_basecase}).
By choosing the size threshold $\geq p$,
we take into account that up to $p-1$ shared vertices are not contracted in
our distributed \Boruvka\ rounds.
As a very last step, we send each identified MST edge back to its original home PE in {\sc RedistributeMST}.

\begin{figure}
	\begin{algorithm}[H]
		\begin{algorithmic}[0]
			\Function{MST}{$G_i = (V_i, E_i)$}
			\State $G_i, T_i \gets$ \Call{localPreprocessing}{$G_i$}
            \While{$\sum{|V_i|} > \mathrm{threshold}$}
			\State $E_i^{\mathrm{min}} \gets$ \Call{minEdges}{$G_i$}
			\State $L^{\mathrm{local}}_i, T_i \gets$ \Call{contractComponents}{$E_i^{\mathrm{min}}, T_i$}
			\State $L^{\mathrm{ghost}}_i \gets$ \Call{exchangeLabels}{$L^{\mathrm{local}}_i, G_i$}
			\State $G_i' \gets$ \Call{relabel}{$L^{\mathrm{local}}_i, L^{\mathrm{ghost}}_i, G_i$}
			\State $G_i \gets$ \Call{redistribute}{$G_i'$}
			\EndWhile
			\State $T_i \gets$ \Call{baseCase}{$G_i, T_i$}
			\State \Return $\Call{redistributeMST}{T_i}$
			\EndFunction
		\end{algorithmic}
		\caption{High-level overview of our distributed \Boruvka{}-MST algorithm.
			By $i$ we denote the rank of a PE.
			The set $T_i$ stores the MST edges.
		}
		\label{alg:basic_msf}
	\end{algorithm}
\end{figure}

\subsection{Local Preprocessing}
\label{subsection_local_kernel}

The key observation behind local preprocessing is that we can contract edges that can be proven to be MST edges
using only locally available information. This results in a graph where the only remaining vertices
have cut edges as lightest incident edges. We implement this using a variant of \Boruvka's algorithm that works on local edges only and only contracts local edges when no lighter cut edge is incident.

After local contraction, we need to update the labels of ghost vertices as
these might have changed. This can be achieved with the label exchange method
described in Section~\ref{subsection_contraction_label_exchange}.
We also have to reestablish the invariant that edges are globally sorted in lexicographic order.
Since we only contract local edges, this can ``almost'' be done by locally resorting the edges.
What needs to be done nonlocally is sorting the edges incident to shared vertices.
This can be achieved very fast for the frequent case that these are short subsequences
allocated to two  subsequent PEs.

\subsection{Component Contraction and Label Exchange}
\label{subsection_contraction_label_exchange}
Recall that the identified minimum incident edges $E_i^{\mathrm{min}}$ define pseudo trees
connecting the components to be contracted.
We first convert them to rooted trees by tie breaking for 2-cycles and by declaring shared vertices as component roots. The rooted trees are then converted to rooted stars by
pointer doubling along the minimum weight edges \cite{chung1996parallel_short}.
This algorithm iteratively halves the depth of the rooted trees by replacing paths of length two by a direct shortcut.
This can be implemented for a tree edge $(u,v)$ by requesting the next edge $(v,w)$ from the
home PE of vertex $v$ and subsequently replacing edge $(u,v)$ by $(u,w)$
An important special case is when $v$ is a shared vertex. This property can be determined locally from the distributed graph data structure. No communication is necessary in this case as $v$ is known to be a component root.
This eliminates a case that would otherwise induce contention at high degree vertices.

When pointer doubling terminates, each PE possesses the labels of the component roots of its local vertices.
In order to obtain the new labels for ghost vertices, for each cut edge
$(u,v)$ the new label of $u$ is sent to the home PE of $(v,u)$.
If there are multiple edges $e$ with the same home PE and $\src{e} = u$, the new
label of $u$ is only sent once.
Note that all requests and replies are implemented with the bulk-operations discussed in \cref{ss:model}.

\subsection{Relabelling and Redistribution}
\label{subsection_redistribution}
In {\sc Relabel}, each PE scans through its edges $(u,v)$ retrieving their labels $u'$ and $v'$ respectively.
If $u'=v'$, the edge is discarded as a self-loop. Otherwise $(u',v')$ is stored as an edge of the contracted graph.

In operation {\sc Redistribute}, the resulting edges are first sorted lexicographically
(depending on the number of edges different sorting algorithms turn out to be useful (see \cref{ss:model}).
Afterwards, edges between the same pair of vertices are consecutive in the sorted result and can be replaced by a single edge with the smallest occurring weight.
Finally, the distributed graph data structure is reestablished using an \allgather-operation on the first edge
on each PE.

\subsection{Base Case Algorithm}
\label{subsection_basecase}

We stop the distributed \Boruvka{} rounds when the remaining number of vertices $n'$ is small enough to be stored on a single PE.
Then we use yet another variant of \Boruvka's
algorithm which was initially proposed by
Adler~\etal~\cite{adler1998communication_short}. Here, the vertices are replicated over all PEs while the edges are distributed arbitrarily, without need to sort
them.  To streamline this, we remap vertex labels to the dense range $1..n'$.
The lightest edge for each vertex can then be computed using an allReduce-operation
with vector length $n'$. Subsequent contraction is then possible by local replicated computations as in the plain \Boruvka{} algorithm (see \cref{subsection_boruvkaMST}).

\subsection{Analysis}\label{ss:analysis}
We essentially port a PRAM algorithm to distributed memory. The former can be proven to run in (expected) time
$\LandauO{\frac{m}{p}\log(n) +\log^2(n)}$ \cite[Theorem~11.8]{sanders2019sequential}.
Using PRAM emulation \cite{ranade91}, we would achieve a bound for distributed memory that is a factor $\log(p)$ larger. Our implementation deviates from this in a way driven by experimental performance evaluation. This reduces overhead at least for benign instances while
accepting asymptotically larger overheads in the worst case. We abstain from a full analysis for reasons of space and because this would yield highly complex formulas.

We note however that the operations involving all edges -- finding locally minimal edges (segmented min-reduction), label exchange (balanced sparse \alltoall{}), and
building the contracted graph (sorting) all use primitives that can avoid the log-factor overhead of the PRAM emulation.

The component contraction using pointer doubling is a more complicated issue.
In the worst case, it needs a logarithmic number of iterations and heavy contention requiring general PRAM emulation can occur. However, pointer doubling only works on vertices rather than on all edges,
in practice a small number of iterations suffices and also contention usually remains manageable.
Thus, we get away with a relatively simple implementation based on sparse \alltoall\ that in our experiments never dominates the running time.

\section{Filter-\Boruvka{}-Algorithm}
\label{section_filter_mst}
In the worst case, \Boruvka's algorithm
looks at all edges a logarithmic number of times.  Asymptotically
better algorithms are known whose running time depends only linearly on the number
of edges $m$.  The most practical of those may still have running time superlinear in $n$
but work very well for sufficiently dense graphs.
To go into that direction also for massively parallel algorithms, we look into
\emph{Filter-Kruskal}\cite{osipov2009filter_short}.
This algorithm is based on the observation that for many instances, most MST edges are quite light.
To exploit this, Filter-Kruskal partitions the edges into light edges $E_{\le}$ and heavy edges $E_{>}$
similar to quicksort. It first recurses on the light edges computing an MSF $T_{\leq}$ of $(V,E_{\le})$. Before also recursing on $E_{>}$,
it \emph{filters} them by eliminating those inside a connected component of $T_{\leq}$.
Filter-Kruskal is also a reasonable shared-memory parallel algorithm \cite{osipov2009filter_short}.
However, Filter-Kruskal has an $\Omega(n)$ critical path length since
MST edges are found sequentially.
To eliminate this bottleneck, we propose \emph{Filter-\Boruvka} which replaces Kruskal's algorithm by \Boruvka's algorithm in the base case.
To facilitate subsequent filtering, we modify the output specification of the underlying \Boruvka~algorithm to
also provide component representatives within the MSF for each vertex.
Changing the base case algorithm
does not affect the performed work but reduces the span from linear to polylogarithmic:

\begin{theorem}
	For random edge weights,
	sequential Filter-\Boruvka\ has expected sequential running time $\LandauO{m+n\log(n)\log(\frac{m}{n}})$.
	When used as a parallel algorithm its expected span is $\LandauO{\log(\frac{m}{n}})$ times the span
	of the underlying parallel \Boruvka\ implementation.
\end{theorem}
\begin{proof}
	The time bound immediately follows from the bound for Filter-Kruskal \cite{osipov2009filter_short}
	since the only difference is in the base case. This does not affect the asymptotic running time since
	both Kruskal's and \Boruvka's algorithm have running time $\LandauO{n\log(n)}$
	for inputs with $\LandauO{n}$ edges.

	Since filtering and partitioning have an (expected) span in $\LandauO{\log(m)}$ and there
	are at most $\LandauO{\log(n)}$ recursion levels prior to a base case
	\Boruvka~call, it suffices to count the expected number of base case \Boruvka~calls to get an upper bound on the overall span.
	To facilitate the analysis, we assume $|E_{\le}|$ and $|E_{>}|$ are in the range $[m/3,2m/3]$.
	If the chosen pivot does not yield such a balanced partition, one can simply repeat
	the pivot selection until it does. This process will need $\LandauO{1}$ repetitions in expectation.
	Using a balanced partitioning, the number of edges which survive filtering is dominated by a random variable $X\sim\mathrm{NB}(n,1/3)$ with negative binomial distribution\cite{karger1995randomized_short}.
	Assume that we stop the recursion when the number of input edges is $\le cn$ with $c > 3$.
	The probability that we recurse on $E_{>}$ then is at most $\mathrm{Pr}[X > cn]$. This is equivalent to the probability
	$\mathrm{Pr}[Y < n]$, where $Y\sim\mathrm{B}(cn,1/3)$ is a binomial distributed random variable.
	The mean of $Y$ is $\mu = cn/3$. Let $\delta = (c-3)/c$.
	Now, we can apply a Chernoff-bound \cite[Theorem~4.5]{mitzenmacher2017probability} to obtain an upper bound on the probability that the right recursion is not stopped directly,
	$\mathrm{Pr}[Y < n] = \mathrm{Pr}[Y < (1-\delta)\mu]\le \exp(-\delta^2\mu/2)$,
    which is in $1/\exp(\Theta(n))$.

	There are $\LandauO{\log(m/n)}$ right recursive calls associated with the leftmost path in the computation tree.
	For each of these right recursive calls, the recursion either stops directly resulting in one additional base case call or continues with an exponentially small probability which yields at most
    $n^{\LandauO{1}}$ additional base case calls.
	Therefore, by the union bound argument, the expected number of base case \Boruvka{}~calls is in
    $\LandauO{\log(m/n) (1 + n^{c_1}/\exp(c_2 n))} = \LandauO{\log(m/n)} + o(1) = \LandauO{\log(m/n)}$ with $c_1$ and  $c_2$ being positive constants.
    \footnote{This is an improved version of the proof in the conference version.}

\end{proof}

Our distributed implementation of Filter-\Boruvka{} is shown in
\cref{alg:filter_msf}. We perform local preprocessing as described in
\cref{subsection_local_kernel} once at the beginning. When the graph is
sufficiently sparse, \ie{} the number of input edges is in $\LandauO{n}$, we stop the recursion and use our distributed \Boruvka-MST
algorithm as base case. There are three differences when calling distributed \Boruvka-MST
to its description given in \cref{alg:basic_msf}:
\begin{itemize}
	\item We refrain from performing local preprocessing as graph locality has
	      already been exploited beforehand.
	\item We do not redistribute the computed MST edges as this is done once at the end of Filter-\Boruvka{}.
	\item There is a distributed array $P$ of size $n$, where PE $i$ holds the elements $in/p..(i+1)n/p$.
	      After a \Boruvka{} round, each PE stores the component root for its local
	      vertices in $P$. In the end, the implicitly constructed trees in $P$ are contracted
	      using $\LandauO{\log(\log(n))}$ pointer doubling rounds (see \cref{subsection_contraction_label_exchange}).
\end{itemize}
For {\sc PivotSelection }, randomly sampled edges are sorted with a distributed sorting algorithm (see \cref{ss:model}).
Afterwards, the median of the sample is broadcasted as pivot element $w_{\mathrm{pivot}}$.

For the filtering step, each PE first requests the labels of the component
representatives for the local vertices in its part of $E_{>}$ from the
distributed array $P$ ({\sc RequestLabels}).
We then use the label exchange and relabelling routines from \cref{subsection_contraction_label_exchange} and
\cref{subsection_redistribution} to rename the edges according to these new labels.
For edges $(u,v)$ that are within one component of the previously computed partial MSF $T_{\leq}$, we now find $u = v$.
Thus, these can be discarded easily. In a last step, we redistribute the edges and eliminate parallel ones as described in \cref{subsection_redistribution}.

\begin{figure}
  \begin{algorithm}[H]
	\begin{algorithmic}[0]
		\Function{Filter-MST}{$G_i = (V_i, E_i)$}
		\State $G_i, T_i \gets$ \Call{localPreprocessing}{$G_i$}
		\State \Call{Rec-Filter-MST}{$G_i, T_i, P$}
		\State \Return $(\Call{redistributeMST}{T_i})$
		\EndFunction
		\Function{Filter}{$G_i = (V_i, E_i), P$}
		\State $L^{\mathrm{local}}_i$ $\gets \Call{requestLabels}{V_i, P}$
		\State $L^{\mathrm{ghost}}_i$ $\gets \Call{exchangeLabels}{L^{\mathrm{local}}_i, G_i}$
		\State $E_i'$ $\gets \Call{relabel}{L^{\mathrm{local}}_i, L^{\mathrm{ghost}}_i, G_i}$
		\State $E_i'' \gets \{ (u,v) \in E_i' \mid u \neq v \}$
		\State \Return $\Call{Redistribute}{(V_i, E_i'')}$
		\EndFunction
		\Function{Rec-Filter-MST}{$G_i = (V_i, E_i), T_i, P$}
		\If{\Call{isSparse}{$G_i, |P|$}}
		\State \Return \Call{MST}{$G_i, P$}
		\EndIf
		\State $w_{\mathrm{pivot}} \gets \Call{PivotSelection}{G_i}$
		\State $E_i^{\le} \gets \Set{(u,v,w) \in E_i}{w \le w_{\mathrm{pivot}}}$
		\State $E_i^{>} \gets \Set{(u,v,w) \in E_i}{w > w_{\mathrm{pivot}}}$
		\State $T_i \gets \Call{Rec-Filter-MST}{(V_i, E_i^{\le}), T_i, P}$
		\State $(V_i', E_i^{>'}) \gets \Call{Filter}{E_i^{>}, P}$
		\State \Return \Call{Rec-Filter-MST}{$(V_i', E_i^{>'}), T_i, P$}
		\EndFunction
	\end{algorithmic}
	\caption{High-level overview of our Filter-\Boruvka{} algorithm. By $i$ we denote the rank of a PE.
		The set $T_i$ stores the MST edges.
		Parameter $P$ denotes a distributed array storing a component representative for each vertex.
	}
	\label{alg:filter_msf}
\end{algorithm}
\end{figure}

\section{Engineering and Implementation Details}\label{s:implementation}
Here, we describe some engineering refinements we made to
obtain a scalable implementation of our algorithms in practice. Our algorithms
are implemented in C++ using MPI for distributed communication in funneled mode, i.e., using only one dedicated thread per process for MPI communication. Multithreading
is realized with OpenMP. For shared-memory parallel algorithmic building blocks
like prefix-sums or filtering, we use the \texttt{parlay} library \cite{blelloch2020parlaylib_short}.
\subsection{Reducing Startup Overhead of All-To-All Exchanges}
Many steps of our two distributed MST algorithms involve multiple \alltoall~exchanges with often only few bytes per message.
During our experiments we noticed that the startup overhead of the built-in MPI routine for the personalized \alltoall{} exchange \texttt{MPI\_Alltoallv} becomes prohibitive with increasing number of PEs.
Therefore, we use an indirect grid based \alltoall~variant if the average number of bytes sent per message is below a certain threshold (we use $500$ on our system).
The PEs are arranged in a (virtual) two dimensional grid
with $c = \floor{\sqrt{p}}$ columns and $r = \ceil{p / c}$ rows.
Note that $c	\le r \le c + 2$. PE $i$ resides in column $\column{i} = i \mod c$ and
$\row{i}	= \floor{i / c}$. Instead of sending messages directly, we exchange them in two
steps. A message from PE $i$ to PE $j$ is first sent to the intermediate PE $t$
in row $\row{j}$ and column $\column{i}$ and from there to its final
destination $j$. Since $t$ is in the same column as $i$ and the same row as
$j$, the message exchanges can be realized with two standard \texttt{MPI\_Alltoallv} exchanges
with at most $\sqrt{p} + 2$ participating PEs, thus reducing the startup
overhead to $\LandauO{\sqrt{p}}$ at the cost of a doubled communication volume. When $p \neq cr$ and $j$ is a member of
the last (incomplete) row of the grid, the intermediate PE is the one in row
$\column{j}$ and column $\column{i}$.
For the second message exchange such a $j$ is (virtually) appended to row $\column{j}$.
For larger $p$, the grid approach can easily be generalized to dimensions $2 < d \le \log(p)$.
For $d = \log(p)$, we basically get the hypercube \alltoall~algorithm from \cite{johnsson1989optimum_short}.

In particular, we found two-level \alltoall~to be crucial for pointer doubling during component contraction. \cref{plot:twolevel_alltoall_effect} shows the accumulated running time of the component contraction phases for an Erd\H{o}s-Renyi graph with $2^{17}$ vertices and $2^{21}$ edges per core.
(One-level-)pointer doubling using the built-in \texttt{MPI\_Alltoallv} routine directly, exhibits significantly increasing running times with a growing number of cores whereas our two-level approach scales very well.
\begin{figure}
	\centering
	\vspace*{-.25cm}
\begin{tikzpicture}[x=1pt,y=1pt]
\definecolor{fillColor}{RGB}{255,255,255}
\path[use as bounding box,fill=fillColor,fill opacity=0.00] (0,0) rectangle (162.61,108.41);
\begin{scope}
\path[clip] (  0.00,  0.00) rectangle (162.61,108.41);
\definecolor{drawColor}{RGB}{255,255,255}
\definecolor{fillColor}{RGB}{255,255,255}

\path[draw=drawColor,line width= 0.6pt,line join=round,line cap=round,fill=fillColor] (  0.00,  0.00) rectangle (162.61,108.41);
\end{scope}
\begin{scope}
\path[clip] ( 28.98, 23.18) rectangle (160.61, 85.38);
\definecolor{fillColor}{RGB}{255,255,255}

\path[fill=fillColor] ( 28.98, 23.18) rectangle (160.61, 85.38);
\definecolor{drawColor}{gray}{0.92}

\path[draw=drawColor,line width= 0.6pt,line join=round] ( 28.98, 26.76) --
	(160.61, 26.76);

\path[draw=drawColor,line width= 0.6pt,line join=round] ( 28.98, 54.66) --
	(160.61, 54.66);

\path[draw=drawColor,line width= 0.6pt,line join=round] ( 28.98, 82.56) --
	(160.61, 82.56);

\path[draw=drawColor,line width= 0.6pt,line join=round] ( 58.90, 23.18) --
	( 58.90, 85.38);

\path[draw=drawColor,line width= 0.6pt,line join=round] (106.76, 23.18) --
	(106.76, 85.38);

\path[draw=drawColor,line width= 0.6pt,line join=round] (154.62, 23.18) --
	(154.62, 85.38);
\definecolor{fillColor}{RGB}{0,158,115}

\path[fill=fillColor] ( 33.54, 25.83) --
	( 36.39, 25.83) --
	( 36.39, 28.69) --
	( 33.54, 28.69) --
	cycle;
\definecolor{fillColor}{RGB}{0,0,0}

\path[fill=fillColor] ( 34.97, 28.23) --
	( 36.89, 24.90) --
	( 33.05, 24.90) --
	cycle;
\definecolor{fillColor}{RGB}{0,158,115}

\path[fill=fillColor] ( 57.47, 31.18) --
	( 60.32, 31.18) --
	( 60.32, 34.04) --
	( 57.47, 34.04) --
	cycle;
\definecolor{fillColor}{RGB}{0,0,0}

\path[fill=fillColor] ( 58.90, 31.89) --
	( 60.82, 28.57) --
	( 56.98, 28.57) --
	cycle;
\definecolor{fillColor}{RGB}{0,158,115}

\path[fill=fillColor] ( 81.40, 40.66) --
	( 84.26, 40.66) --
	( 84.26, 43.52) --
	( 81.40, 43.52) --
	cycle;
\definecolor{fillColor}{RGB}{0,0,0}

\path[fill=fillColor] ( 82.83, 33.23) --
	( 84.75, 29.90) --
	( 80.91, 29.90) --
	cycle;
\definecolor{fillColor}{RGB}{0,158,115}

\path[fill=fillColor] (105.33, 51.99) --
	(108.19, 51.99) --
	(108.19, 54.85) --
	(105.33, 54.85) --
	cycle;
\definecolor{fillColor}{RGB}{0,0,0}

\path[fill=fillColor] (106.76, 34.33) --
	(108.68, 31.00) --
	(104.84, 31.00) --
	cycle;
\definecolor{fillColor}{RGB}{0,158,115}

\path[fill=fillColor] (129.27, 65.51) --
	(132.12, 65.51) --
	(132.12, 68.37) --
	(129.27, 68.37) --
	cycle;
\definecolor{fillColor}{RGB}{0,0,0}

\path[fill=fillColor] (130.69, 35.83) --
	(132.61, 32.51) --
	(128.77, 32.51) --
	cycle;
\definecolor{fillColor}{RGB}{0,158,115}

\path[fill=fillColor] (153.20, 77.85) --
	(156.05, 77.85) --
	(156.05, 80.71) --
	(153.20, 80.71) --
	cycle;
\definecolor{fillColor}{RGB}{0,0,0}

\path[fill=fillColor] (154.62, 37.88) --
	(156.55, 34.55) --
	(152.70, 34.55) --
	cycle;
\definecolor{drawColor}{RGB}{0,158,115}

\path[draw=drawColor,line width= 0.6pt,dash pattern=on 1pt off 3pt ,line join=round] ( 34.97, 27.26) --
	( 58.90, 32.61) --
	( 82.83, 42.09) --
	(106.76, 53.42) --
	(130.69, 66.94) --
	(154.62, 79.28);
\definecolor{drawColor}{RGB}{0,0,0}

\path[draw=drawColor,line width= 0.6pt,dash pattern=on 1pt off 3pt ,line join=round] ( 34.97, 26.01) --
	( 58.90, 29.68) --
	( 82.83, 31.01) --
	(106.76, 32.11) --
	(130.69, 33.61) --
	(154.62, 35.66);
\definecolor{drawColor}{gray}{0.20}

\path[draw=drawColor,line width= 0.6pt,line join=round,line cap=round] ( 28.98, 23.18) rectangle (160.61, 85.38);
\end{scope}
\begin{scope}
\path[clip] ( 28.98, 85.38) rectangle (160.61,101.25);
\definecolor{drawColor}{RGB}{255,255,255}
\definecolor{fillColor}{RGB}{255,255,255}

\path[draw=drawColor,line width= 0.6pt,line join=round,line cap=round,fill=fillColor] ( 28.98, 85.38) rectangle (160.61,101.25);
\definecolor{drawColor}{RGB}{0,0,0}

\node[text=drawColor,anchor=base,inner sep=0pt, outer sep=0pt, scale=  0.80] at ( 94.80, 90.56) {GNM $(2^{17}, 2^{21})$};
\end{scope}
\begin{scope}
\path[clip] (  0.00,  0.00) rectangle (162.61,108.41);
\definecolor{drawColor}{gray}{0.20}

\path[draw=drawColor,line width= 0.6pt,line join=round] ( 58.90, 20.43) --
	( 58.90, 23.18);

\path[draw=drawColor,line width= 0.6pt,line join=round] (106.76, 20.43) --
	(106.76, 23.18);

\path[draw=drawColor,line width= 0.6pt,line join=round] (154.62, 20.43) --
	(154.62, 23.18);
\end{scope}
\begin{scope}
\path[clip] (  0.00,  0.00) rectangle (162.61,108.41);
\definecolor{drawColor}{RGB}{0,0,0}

\node[text=drawColor,anchor=base west,inner sep=0pt, outer sep=0pt, scale=  0.80] at ( 54.10, 11.37) {2};

\node[text=drawColor,anchor=base west,inner sep=0pt, outer sep=0pt, scale=  0.56] at ( 58.10, 14.64) {10};

\node[text=drawColor,anchor=base west,inner sep=0pt, outer sep=0pt, scale=  0.80] at (101.96, 11.37) {2};

\node[text=drawColor,anchor=base west,inner sep=0pt, outer sep=0pt, scale=  0.56] at (105.96, 14.64) {12};

\node[text=drawColor,anchor=base west,inner sep=0pt, outer sep=0pt, scale=  0.80] at (149.83, 11.37) {2};

\node[text=drawColor,anchor=base west,inner sep=0pt, outer sep=0pt, scale=  0.56] at (153.82, 14.64) {14};
\end{scope}
\begin{scope}
\path[clip] (  0.00,  0.00) rectangle (162.61,108.41);
\definecolor{drawColor}{RGB}{0,0,0}

\node[text=drawColor,anchor=base east,inner sep=0pt, outer sep=0pt, scale=  0.80] at ( 24.03, 24.01) {0.1};

\node[text=drawColor,anchor=base east,inner sep=0pt, outer sep=0pt, scale=  0.80] at ( 24.03, 51.90) {1.0};

\node[text=drawColor,anchor=base east,inner sep=0pt, outer sep=0pt, scale=  0.80] at ( 24.03, 79.80) {10.0};
\end{scope}
\begin{scope}
\path[clip] (  0.00,  0.00) rectangle (162.61,108.41);
\definecolor{drawColor}{gray}{0.20}

\path[draw=drawColor,line width= 0.6pt,line join=round] ( 26.23, 26.76) --
	( 28.98, 26.76);

\path[draw=drawColor,line width= 0.6pt,line join=round] ( 26.23, 54.66) --
	( 28.98, 54.66);

\path[draw=drawColor,line width= 0.6pt,line join=round] ( 26.23, 82.56) --
	( 28.98, 82.56);
\end{scope}
\begin{scope}
\path[clip] (  0.00,  0.00) rectangle (162.61,108.41);
\definecolor{drawColor}{RGB}{0,0,0}

\node[text=drawColor,anchor=base,inner sep=0pt, outer sep=0pt, scale=  0.80] at ( 94.80,  1.56) {number of cores};
\end{scope}
\begin{scope}
\path[clip] (  0.00,  0.00) rectangle (162.61,108.41);
\definecolor{drawColor}{RGB}{0,0,0}

\node[text=drawColor,rotate= 90.00,anchor=base,inner sep=0pt, outer sep=0pt, scale=  0.80] at (  5.51, 54.28) {running time (sec)};
\end{scope}
\begin{scope}
\path[clip] (  0.00,  0.00) rectangle (162.61,108.41);
\definecolor{drawColor}{RGB}{190,190,190}

\path[draw=drawColor,line width= 0.4pt,line join=round,line cap=round] ( 35.87, 67.70) rectangle ( 82.65, 81.92);
\end{scope}
\begin{scope}
\path[clip] (  0.00,  0.00) rectangle (162.61,108.41);
\definecolor{fillColor}{RGB}{255,255,255}

\path[fill=fillColor] ( 35.87, 67.70) rectangle ( 82.65, 81.92);
\end{scope}
\begin{scope}
\path[clip] (  0.00,  0.00) rectangle (162.61,108.41);
\definecolor{fillColor}{RGB}{255,255,255}

\path[fill=fillColor] ( 35.87, 74.81) rectangle ( 50.32, 81.92);
\end{scope}
\begin{scope}
\path[clip] (  0.00,  0.00) rectangle (162.61,108.41);
\definecolor{fillColor}{RGB}{0,158,115}

\path[fill=fillColor] ( 41.67, 76.94) --
	( 44.52, 76.94) --
	( 44.52, 79.79) --
	( 41.67, 79.79) --
	cycle;
\end{scope}
\begin{scope}
\path[clip] (  0.00,  0.00) rectangle (162.61,108.41);
\definecolor{drawColor}{RGB}{0,158,115}

\path[draw=drawColor,line width= 0.6pt,dash pattern=on 1pt off 3pt ,line join=round] ( 37.31, 78.37) -- ( 48.88, 78.37);
\end{scope}
\begin{scope}
\path[clip] (  0.00,  0.00) rectangle (162.61,108.41);
\definecolor{fillColor}{RGB}{255,255,255}

\path[fill=fillColor] ( 35.87, 67.70) rectangle ( 50.32, 74.81);
\end{scope}
\begin{scope}
\path[clip] (  0.00,  0.00) rectangle (162.61,108.41);
\definecolor{fillColor}{RGB}{0,0,0}

\path[fill=fillColor] ( 43.09, 73.47) --
	( 45.02, 70.14) --
	( 41.17, 70.14) --
	cycle;
\end{scope}
\begin{scope}
\path[clip] (  0.00,  0.00) rectangle (162.61,108.41);
\definecolor{drawColor}{RGB}{0,0,0}

\path[draw=drawColor,line width= 0.6pt,dash pattern=on 1pt off 3pt ,line join=round] ( 37.31, 71.25) -- ( 48.88, 71.25);
\end{scope}
\begin{scope}
\path[clip] (  0.00,  0.00) rectangle (162.61,108.41);
\definecolor{drawColor}{RGB}{0,0,0}

\node[text=drawColor,anchor=base west,inner sep=0pt, outer sep=0pt, scale=  0.70] at ( 55.82, 75.96) {one-level};
\end{scope}
\begin{scope}
\path[clip] (  0.00,  0.00) rectangle (162.61,108.41);
\definecolor{drawColor}{RGB}{0,0,0}

\node[text=drawColor,anchor=base west,inner sep=0pt, outer sep=0pt, scale=  0.70] at ( 55.82, 68.84) {two-level};
\end{scope}
\end{tikzpicture}
	\vspace*{-0.25cm}
	\caption{Effect of two-level \alltoall{}~on the component contraction of \cref{alg:basic_msf}.}
	\label{plot:twolevel_alltoall_effect}
\end{figure}
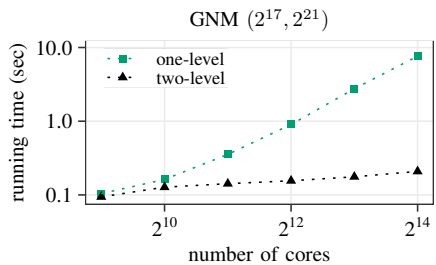

\subsection{Local Preprocessing Optimizations}
We enhance the modified \Boruvka{} variant used for the local preprocessing step with a recursive edge-filtering approach as in Filter-\Boruvka{}. Our multithreaded implementation uses the \texttt{Min-Priority-Write}-approach for minimum edge computation (as well as some other building blocks)
from a fast shared-memory MST algorithm \cite{dhulipala2021theoretically_short}.
We apply the preprocessing only if at least 10\% of the edges are local.

In preliminary experiments, we found that local preprocessing considerably reduces the number of vertices leaving many parallel edges.
Instead of sorting all edges $E$ to remove parallel ones, we determine a pivot weight $w$ such that the set $E'$ of edges lighter than $w$ is small.
The edges from $E'$ are then inserted into a hash table $H$ omitting their weight. In a subsequent scan over $E$ we can filter out all edges that are in $H$. We then only have to sort the remaining edges and remove parallel ones therein in a final scan.
This variant outperforms the pure sorting approach by up to a factor of 2.5 if the hash table remains small enough to fit into the cache.
\subsection{Further Remarks}
Regarding distributed sorting we use distributed hypercube quicksort \cite{axtmann2017robust_short} if the average number of elements to sort per PE is below $512$.
For larger inputs we use our own implementation of distributed two-level sample sort (similar to AMS-sort \cite{axtmann2015practical_short,axtmann2017robust_short}) applying the hypercube algorithm to sort the samples.

In order to be able to output the original source and destination vertices of an MST edge, we add an id to every edge prior to the actual MST computation. The id of an MST edge is then looked up in a copy of the initial edge list.
As main memory on compute cluster nodes is notoriously scarce, this copy is stored with $7$-bit variable length encoding on the differences of consecutive vertices.
Note that the time for encoding is not included in our experiments, however, we account for decoding the compressed edge list twice (before and after the actual MST computation).%
\footnote{Also note that encoding takes at most 30\% of the running time and that our competitors in \cref{s:experiments} do not produce a similarly prepared output.}

For switching to the base case in our distributed \Boruvka{}~algorithm, we use
the maximum of two times the number of MPI processes and $35\,000$. In
distributed Filter-\Boruvka{}, we stop the recursive partitioning and employ
our distributed \Boruvka~algorithm as base case when the average degree
of the graph is four or less. We also refrain from further partitioning if the graph
is too small (less than 1000 edges per MPI processes). Furthermore, if the
number of filtered edges is too small, these are not processed directly but
propagated back to the previous recursion level and merged with the heavy edges
there.

\section{Experiments}\label{s:experiments}
We now discuss the experimental evaluation of our algorithms. An implementation is available at \url{https://github.com/mschimek/kamsta}.
We compare our algorithms (\texttt{boruvka} and \texttt{filterBoruvka}) against two state-of-the-art competitors: \texttt{sparseMatrix} by Baer~\etal~\cite{baer2022parallel_short} and \texttt{MND-MST} by Panja~\etal~\cite{panja2018mnd_short}.
All algorithms are written in \texttt{C++} and compiled with \texttt{g++} version 10.2.0 using optimizations \texttt{-O3} and \texttt{-march=native}.
We use \texttt{OpenMPI} version 4.0.4 for interprocess communication. All algorithms use multithreading with \texttt{OpenMP}.

\begin{table}[t]
	\centering
	\caption{Real-world instances used in our strong scaling experiments.}
	\begin{tabular}{
			l
			S[table-format=3.1e1, round-precision = 1, round-mode=places, scientific-notation=engineering]
			S[table-format=3.1e1, round-precision = 1, round-mode=places, scientific-notation=engineering]
			l
			l
		}
		Graph        & $n$        & $m$          & Source                                     & Type                    \\
		\hline
		\friendsterN & 68349466   & 3623698684   & KONECT~\cite{konect_short}                 & \multirow{2}{*}{social} \\
		\twitterN    & 41652230   & 2405026092   & ~\cite{twitter_short}                      &                         \\
		\hline
		\ukwebN      & 105896436  & 6603753128   & LAW~\cite{BoVWFI_short,BRSLLP_short}       & \multirow{3}{*}{web}    \\
		\itwebN      & 41291594   & 2054949894   & LAW~\cite{BoVWFI_short,BRSLLP_short}       &                         \\
		\wcdN        & 1724573718 & 123870780939 & WDC~\cite{meusel2015graph_short}           &                         \\
		\hline
		\usroadN     & 23947347   & 57708624     & DIMACS~\cite{demetrescu2009shortest_short} & road                    \\
	\end{tabular}
	\vspace*{0.1cm}
	\label{tab:instances}
	\vspace*{-0.45cm}
\end{table}

\begin{figure*}
	\centering
	\vspace*{-.50cm}
	\input{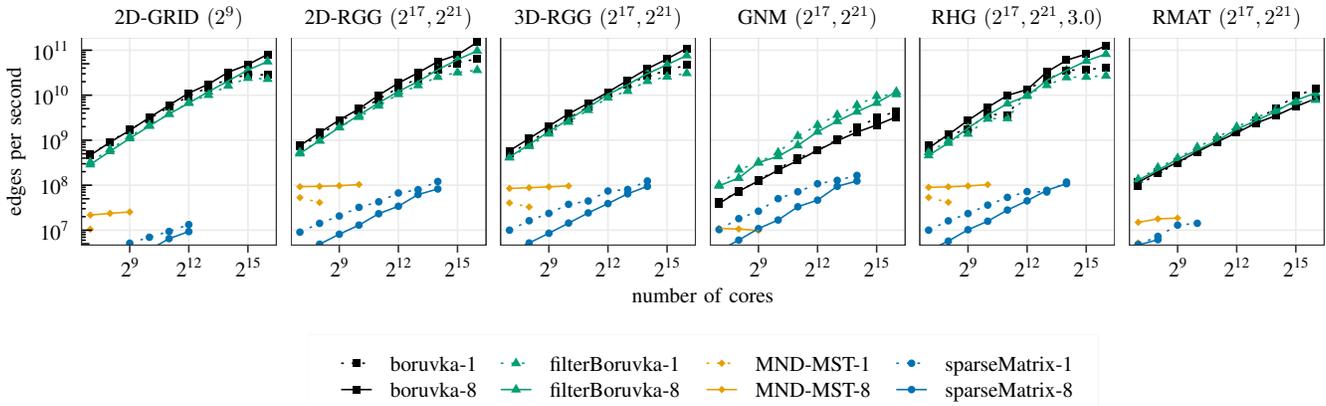}
	\vspace*{-0.45cm}
	\caption{Throughputs in our weak scaling experiments with $2^{17}$ vertices and $2^{21}$ edges per core. Each algorithm was executed with one and eight threads.}
	\label{plot:synthetic_17_20}
\end{figure*}
\texttt{SparseMatrix} adapts the Awerbuch-Shiloach PRAM algorithm\cite{awerbuch1987new_short} to the distributed setting using linear algebra primitives and leveraging Cyclops, a library for generalized sparse tensor algebra.
They use a 2D-partitioning scheme where the graph's adjacency matrix is divided into blocks which are then distributed among the PEs.
The source code is publicly available\footnote{\url{https://github.com/raghavendrak/algebraic MSF}}.
\texttt{MND-MST} is a multinode GPU-CPU algorithm, which also comprises a CPU-only version. It uses \Boruvka{}'s algorithm to compute local MST edges and to contract the incident vertices.
Afterwards, fixed size groups of PEs exchange parts of the previously contracted vertices and iteratively apply \Boruvka{}'s algorithm on their local input.
Once a threshold on the size of the reduced graph is reached, all group members send their contracted graphs to the group \emph{leader}.
Then, the whole process starts again with only the group leaders performing computations.
As in our algorithms, they use 1D-partitioning. However, they do not share vertices beyond process boundaries which can lead to load imbalances for graphs with very skewed degree distributions.
The source code was provided by the authors directly. However, it seems to deviate somewhat from the algorithm described in the paper \cite{panja2018mnd_short}.

For our evaluation, we mainly\footnote{We also conducted experiments on the smaller HoreKa supercomputer. As the
results obtained there are in line with our findings from SuperMUC-NG, we omit
them due to space limitations.} use the SuperMUC-NG supercomputer. The thin node cluster consists of 6\,336 compute nodes with a total of 304\,128 cores.
A compute nodes is equipped with an Intel Xeon Platinum 8174 processor with 48 cores and 96 GByte of main memory and runs SUSE Linux Enterprise Server 15 SP3.
The nodes are connected via an OmniPath network with 100 Gbit/s bandwidth.

\paragraph*{Instances and Methodology}
For our strong scaling experiments, we use the six real-world graphs listed in \cref{tab:instances}.
The number of (symmetric, directed) edges of these graphs ranges from 57 millions to 123 billions.
In our weak scaling experiments, we use instances of six different graph families generated with \texttt{KaGen}\footnote{\url{https://github.com/sebalamm/KaGen}} \cite{funke2019communication_short} and a fast RMAT generator \cite{hubschle2020linear_short}.
We use two-dimensional grids (2D-GRID), two- and three-dimensional random geometric graphs (2D/3D-RGG), random hyperbolic graphs (RHG), Erd\H{o}s-Renyi graphs (GNM) and RMAT graphs.

RGGs are constructed by placing vertices uniformly at random in the unit square (unit cube for 3D) and each MPI process is assigned a part thereof.
Vertices are connected if the Euclidean distance is below a threshold $d$.
RHG construction is conceptually similar, as vertices are placed on a disk with radius $r$ that depends on the average degree and power-law exponent $\gamma$, where the disk is again evenly divided among the MPI processes.
Two vertices are adjacent, if the (hyperbolic) distance is smaller than $r$.
In Erd\H{o}s-Renyi graphs, each edge is inserted with a probability given as an input parameter. RMAT graphs are generated by recursively partition the $n\times n$ adjacency matrix. An edge is inserted as soon as a $1\times1$ partition size is reached. We use the default probabilities from the Graph500 benchmark.

Grids and RGGs
have a high degree of locality.
GNMs and RMAT graphs consist almost exclusively of cut-edges.
RHGs are somewhere in between.
RHGs (for which we use power-law exponent $\gamma=3.0$) and RMAT graphs possess a power-law degree distribution.

All graphs are scaled such that the number of vertices and number of (symmetric, directed) edges are proportional to the number of cores used ($=\mathrm{\#MPIprocesses} \times \mathrm{\#threads}$). For RGG/GNM the threshold distance/edge probability is chosen accordingly.
\texttt{KaGen} ensures that the generated edges are globally lexicographically sorted and thus do not produce shared vertices for the input.
Regarding the RMAT generator \cite{hubschle2020linear_short}, we first globally sort the generated edges and then redistribute them equally over all PEs.
For \texttt{MND-MST}, the edges incident to a shared vertex are moved completely to one MPI process to meet their input format.
Following the experimental setup in \cite{baer2022parallel_short}, we assign a weight drawn uniformly at random from $[1,255)$
to each edge.
On every benchmark set, all algorithms are run at least three times with one \emph{warm-up} round which we discard. As the variances in running time are usually very low, we report only the mean running time.

\subsection{Weak Scaling}
\cref{plot:synthetic_17_20} shows the throughput (measured in edges per second)
achieved in our weak scaling experiments with $2^{17}$ vertices and $2^{21}$
edges per core on up to $2^{16}$ cores. Due to the high running time we ran our
competitors only up to $2^{14}$ cores to save computation time. \texttt{MND-MST} crashed beyond $1024$
cores on all instances. On the grid and RMAT instances,
\texttt{sparseMatrix} also crashed beyond $4096$ and $1024$ respectively.
Our two algorithms (\texttt{boruvka} and \texttt{filterBoruvka}) outperform the competitors clearly on all instances.

Especially on graphs
with high locality, we achieve speedups of two orders of magnitude
over our competitors peaking at a factor $800$ for grid graphs.
For \texttt{sparseMatrix} we attribute this to the fact that the algorithm does not exploit locality. Furthermore, their 2D-partitioning makes exploiting graph locality more challenging from a structural point of view as only the processors on the diagonal of the matrix possess local edges.
\texttt{MND-MST} is faster on these inputs as its design aims at (and relies on) using locality in graphs. We believe that its scalability problems are to some extent due to weaknesses in its implementation which is not designed to scale to several thousands of cores.

For GNM and RMAT, local preprocessing is not effective as there are only very few local edges. Nevertheless, we are up to $36$ times faster than our competitors. Moreover,
we see -- especially for GNM -- the effectiveness of our filter approach being
up to $4$ times faster than our non-filter variant. In additional weak scaling
experiments on denser graphs with $2^{23}$ edges per core, which we omit due to space limitations, this effect is even stronger.
For graphs with high locality, our filter approach performs slightly worse since the locally contracted graphs are significantly smaller than the input and the overhead introduced by the filtering step does not pay off.

Our 8-thread variants are faster and scale better than their 1-thread counterparts for highly local graphs. This is as expected as each MPI process obtains a larger part of the graph which enables a more effective local contraction yielding a smaller remaining graph.
Surprisingly, the 1-thread variant is noticeably faster for GNM (and on some configurations for RMAT). More detailed measurements reveal that one reason for this is the time spent within \texttt{MPI\_Alltoallv} exchanges. Here, we seem to be limited by the single-threaded execution within MPI.

Overall, our algorithms prove to be scalable on up to $2^{16}$ cores even on RMAT graphs that are notorious for their highly skewed degree distribution which often imposes scalability problems.
We attribute our good scalability on RMAT instances to our edge-based 1D-partitioning as it implicitly splits (very) high degree vertices (and their incident edges) in shared vertices over multiple processors preventing load imbalances.

\cref{plot:local_contraction_effect} shows our algorithms with disabled local preprocessing on grids and random geometric/hyperbolic graphs with $2^{17}$ vertices and $2^{23}$ edges per core.
We can see that local contraction makes our algorithms up to $5$ times faster. Moreover, it turns out that the filtering approach is also beneficial for graphs with many local edges if the graph size is not too small.
This indicates that on machines with more memory per compute node, which would enable us to tackle larger graphs, filtering could also be beneficial on these graph families with local preprocessing enabled.

\cref{plot:runtime_ratio} shows the normalized running time distribution of the different phases of our four variants (\textbf{b}oruvka-\{1,8\}, \textbf{f}ilterBoruvka-\{1,8\}). We see that for 3D-RGG (prototypical for the other graphs with high locality) a considerable amount of time is consumed by the local preprocessing.
For GNM and RMAT, the local preprocessing time is negligible as these graphs possess many cut-edges and we skip this step after a quick check if the number of cut-edges exceeds $90\%$.
For these two graph families, most of the running time is spent in label exchange and the redistribution of the edges.
The running time distribution of Filter-\Boruvka~shows that the time spent in these communication intense phases can be significantly reduced with filtering, which in turn becomes dominant for GNM and RMAT.
We further see that by successfully applying our two-level \alltoall, the time spent for pointer doubling during component contraction does only contribute a minor factor to the running time for all graphs. This also justifies our focus on per-edge computation in the analysis (\cref{ss:analysis}).

\begin{figure}
	\centering
	\vspace*{-.45cm}
	\input{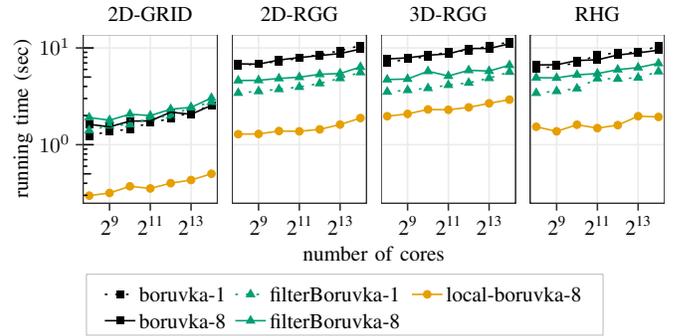}
	\vspace*{-0.45cm}
	\caption{Running time of our algorithms without local preprocessing on highly-local graphs with $2^{17}$ vertices and $2^{23}$ edges per core. Our fastest variant with local preprocessing enabled -- \texttt{local-boruvka-8} -- is given as a baseline.}
	\label{plot:local_contraction_effect}
\end{figure}

\begin{figure*}
	\centering
	\vspace*{-.50cm}
	\input{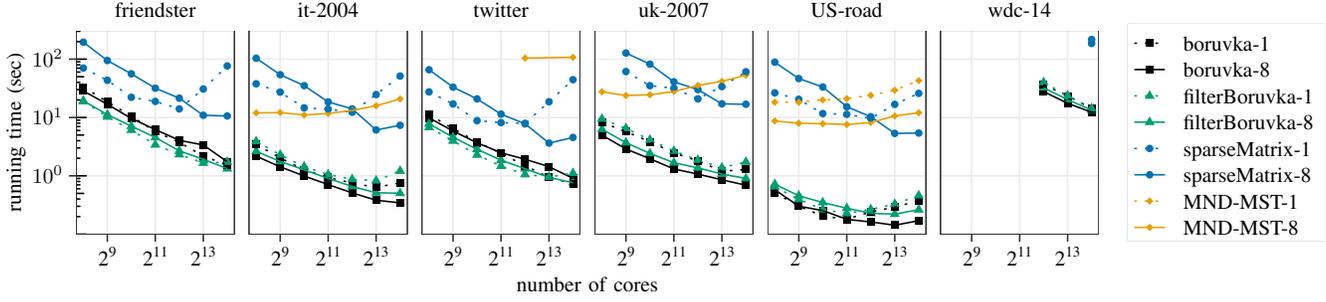}
	\vspace*{-0.45cm}
	\caption{Strong scaling experiments on real-world graphs.}
	\label{plot:real_world}
\end{figure*}

\begin{figure*}
	\centering
	\vspace*{-.45cm}
	\input{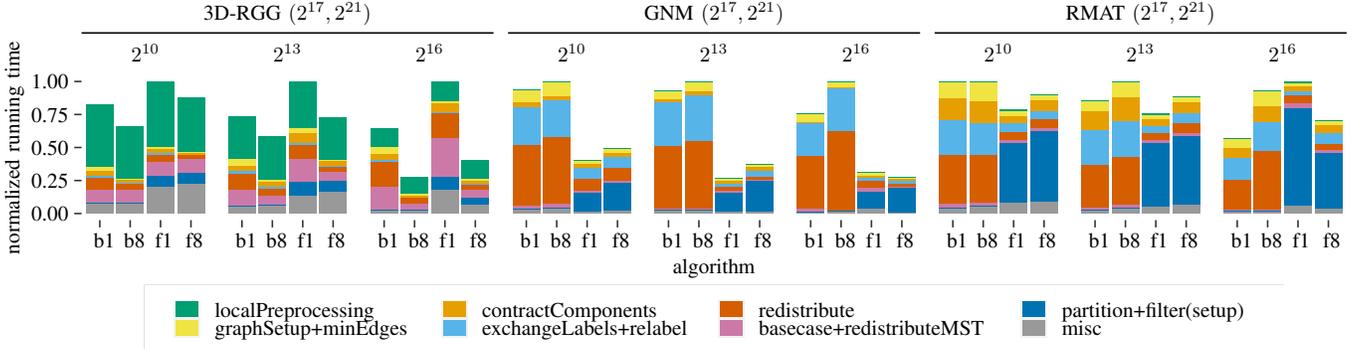}
	\vspace*{-0.40cm}
	\caption{Normalized running times to the range $[0,1]$ of different steps of our algorithms with respect to the slowest variant in each graph$\times$number-of-cores configuration.}
	\label{plot:runtime_ratio}
\end{figure*}

\subsection{Strong Scaling}
We now discuss the strong scaling experiments presented in \cref{plot:real_world}. Due to memory limitations, our competitors could not process all graphs on all configurations.
Our own algorithms are able to process all graphs on $2^8$ to $2^{14}$ processors except for \wcd{} for which we also need at least $4096$ cores.
Our algorithms exhibit good scalability and are $4$ to $40$ times faster than our competitors, which also scale worse for all graphs but \usroad{}.
For \usroad{}, which is relatively small with only 57 million edges, we achieve our best running time for $8192$ cores.

Furthermore, we see that the 8-thread variants tend to outperform their 1-thread counterparts with increasing number of cores used.
This is not surprising as the number of vertices and edges per MPI process decreases and the latency-overhead introduced is less and less compensated by local work and communication volume.

For the social instances, our filtering approach tends to be faster than our non-filter algorithm. For all other graphs, our non-filter approach performs better.

\subsection{Comparisons with Shared-Memory Algorithms}
For their algorithm \texttt{MASTIFF}, Esfahani~\etal~\cite{esfahani2022mastiff_short} provide measurements on a $128$-core shared-memory server with $2$ TB main memory for \twitter, \friendster{}, \usroad{} and \wcd{}.
Comparing their running times for the first three graphs with the fastest of our algorithms for each graph on $256$ cores ($6$ compute nodes with a total of $576$ GB main memory) yields an average speedup of \texttt{MASTIFF} over our algorithms of $2.5$.
From $1024$ cores on, we are faster on \friendster{} and \usroad. For \twitter{}, we need $2048$ cores.
Due to memory limitations, we need $4096$ cores ($86$ compute nodes with $8.3$ TB main memory) to be able to process \wcd{} in $27.8$ seconds while \texttt{MASTIFF} processes this graph in $45.7$ seconds on their machine.
Although the evaluation setting is not identical, this rough comparison indicates that our algorithms are only a modest factor slower than state-of-the-art shared-memory algorithms.

\section{Conclusion and Future Work}\label{s:conclusion}
We have demonstrated that MSTs on huge networks can be computed on
large supercomputers in a scalable way.
Improvements over previous approaches are typically one or two orders of magnitude, approaching three orders of magnitude for some large configurations. We achieve that by maximizing local computations and by using
efficient primitives for global communication like two-phase sparse \alltoall{}~and scalable sorters.
We see this as an important step in a larger
effort to obtain efficient massively parallel graph algorithms on a larger range
of problems.

While this was generally successful, our results also demonstrate several
challenges that may also apply to other graph problems. Perhaps most obviously,
global communication remains the main bottleneck which might be mitigated by
further refining the used primitives and by devising new algorithms that use
less global synchronization. From a more theoretical perspective we also want
algorithms that reduce the gap between theory and practice, \eg{} by giving
practical implementations with asymptotic worst-case performance at least as
good as using PRAM emulation.

More surprisingly, we observe a complicated tradeoff regarding the most
efficient number of threads per process. Arriving at an implementation of the
shared-memory parts that efficiently use all hardware threads of a compute node
would significantly improve also overall scalability since it enables more
local contraction and more coarse-grained communication. Since similar effects
show up in other problems we study, this seems to be an important and
nontrivial area of further investigation.

From a more practical perspective, our approach of a sophisticated, relatively
low-level implementation for one of the most basic graph problems makes sense
for basic research on parallel algorithms but is untenable for more complex
problems. Transferring our techniques in a generic way into more high-level
tools like sparse matrix libraries or graph tools therefore seems an important
direction of further research in order to reduce the observed performance and
scalability gap.

Interestingly, single \Boruvka{} rounds are also an
important part of more sophisticated MST algorithms with better performance
guarantees like the expected linear time algorithm\cite{karger1995randomized_short} and the related PRAM
algorithm\cite{ColeKT96_short}. Therefore, we believe that the algorithmic
building blocks developed in this work can also be of interest for distributed
implementations of such more complex MST algorithms.

\section*{Acknowledgment}
The authors gratefully acknowledge the Gauss Centre for Supercomputing e.V.
(www.gauss-centre.eu) for funding this project by providing computing time on
the GCS Supercomputer SuperMUC-NG at Leibniz Supercomputing Centre
(www.lrz.de).
This work was performed on the HoreKa supercomputer funded by the
Ministry of Science, Research and the Arts Baden-Württemberg and by
the Federal Ministry of Education and Research.
\bibliography{references}
\end{document}